\documentclass[12pt]{amsart}

\usepackage{amsmath,amssymb,amsbsy,amsfonts,amsthm,latexsym,
                     amsopn,amstext,amsxtra,euscript,amscd,mathrsfs}
\usepackage{amsmath,amssymb,amsbsy,amsfonts,latexsym,amsopn,amstext,cite,
                                               amsxtra,euscript,amscd,bm}
\usepackage{mathtools}
\usepackage{todonotes}
\usepackage{url}
\usepackage[colorlinks,linkcolor=blue,anchorcolor=blue,citecolor=blue]{hyperref}

\newcommand{\F}{\mathbb{F}}
\newcommand{\Q}{\mathbb{Q}}

\newcommand{\Z}{\mathbb{Z}}

\def\cA{{\mathcal A}}
\def\cB{{\mathcal B}}
\def\cC{{\mathcal C}}

\def\cL{{\mathcal L}}

\def\cS{{\mathcal S}}

\def\({\left(}
\def\){\right)}
\def\[{\left[}
\def\]{\right]}
\def\<{\langle}
\def\>{\rangle}
\def\fl#1{\left\lfloor#1\right\rfloor}
\def\rf#1{\left\lceil#1\right\rceil}
\def\mand{\qquad\mbox{and}\qquad}

\def\vec#1{\mathbf{#1}}

\newtheorem{theorem}{Theorem}

\newtheorem{example}[theorem]{Example}
\newtheorem{rem}[theorem]{Remark}
\newtheorem{corollary}[theorem]{Corollary}
\newtheorem{algol}[theorem]{Algorithm}

\numberwithin{equation}{section}
\numberwithin{theorem}{section}

\begin{document}

\title{Codes correcting restricted errors}

\author{Igor E. Shparlinski}
\address{School of Mathematics and Statistics, University of New South Wales,
 Sydney NSW 2052, Australia}
\email{igor.shparlinski@unsw.edu.au}

\author{Arne Winterhof}

\address{Johann Radon Institute for Computational and Applied Mathematics, Austrian Academy of Sciences, Altenberger Str.\ 69, Linz, Austria}
\email{arne.winterhof@oeaw.ac.at}

\begin{abstract}
 We study the largest possible length $B$ of $(B-1)$-dimensional linear codes over $\F_q$ which can correct up to $t$ errors taken from a restricted set $\cA\subseteq \F_q^*$.
 Such codes can be applied to multilevel flash memories. 
 
 Moreover, in the case that $q=p$ is a prime and the errors are limited by a constant we show that often the primitive $\ell$th roots of unity, where $\ell$ is a prime divisor of $p-1$, 
 define good such codes. 
\end{abstract}


\maketitle

\section{Introduction}

Let $\F_q$ be the finite field of $q$ elements, $\cA$ be a nonempty subset of $\F_q^*=\F_q\setminus\{0\}$ 
and $t$ be a positive integer.
We call a subset $\cB\not=\emptyset$ of $\F_q^*$ 
of size $B= \# \cB$   a {\em $(t,\cA,q)$-packing set} if 
for any $x\in \F_q$ there is at most one solution 
$$
\vec{a}=(a_b)_{b\in \cB}\in (\cA\cup \{0\})^{B}
$$ 
to the equation in $\F_q$ 
$$\sum_{b\in \cB} a_b b=x
$$
with Hamming-weight $w(\vec{a})\le t$, that is, $\vec{a}$ has at most $t$ nonzero coordinates.

We can use the elements of $\cB$ to define a $(B-1)$-dimensional linear code $\cC$ of length $B$ with the one line parity check matrix $H=(b)_{b\in \cB}$:
$$
\cC=\left\{(c_b)_{b\in \cB}\in \F_q^{B} : \sum_{b\in \cB} c_bb=0\right\}.$$
Using nearest neighbor decoding without further knowledge about the errors such a code of minimum weight at most  $2$ (by the Singleton bound)
cannot correct any error, see for example \cite{be,lixi,li,mc}. However, if we assume that all occurring errors $a$ are elements of $\cA$, we can correct up to $t$ errors.
More precisely, the syndromes
$$S_{\cB}(\vec{a})=\sum_{b\in \cB} a_bb$$
of any error  
$\vec{a}=(a_b)_{b\in \cB}\in (\cA\cup \{0\})^{B}$ of Hamming weight at most $t$ are all distinct and we can determine the unique error $\vec{a}$
from a syndrome table since $S_{\cB}(\vec{c}+\vec{a})=S_{\cB}(\vec{a})$ for every code word ${\bf c}\in \cC$ uniquely determines~$\vec{a}$. 
Note that the information rate $(B-1)/B$ of these codes improves with the size $B$ of $\cB$ and we are interested in $(t,\cA,q)$-packing sets of large size.

A particularly interesting case is 
\begin{equation}\label{eq:special}
 \cA=\{1,2,\ldots,\lambda\}\quad \mbox{and }\quad \mbox{$q=p$ is a prime},
\end{equation}
that is, we can correct all errors of limited magnitude $\lambda$, where $\lambda$ is a positive integer. Such codes have been first 
proposed in~\cite{ah} which in turn is based on~ \cite{var}. 
They are used for multilevel flash memories, 
see~\cite{ca,kl,kl2}.

In Section~\ref{sec:upper} we present an upper bound on the size of any $(t,\cA,q)$-packing set $\cB$ with any 
nonempty  set  $\cA \subseteq \F_q^*$, which is a simple generalization  of the upper bound of~\cite{kl} for the special case~\eqref{eq:special}. 
We also provide examples of sets $\cA$ for which this bound is tight.

In Section~\ref{sec:lower} we prove a lower bound on the size of any {\em maximal} $(t,\cA,q)$-packing set $\cB$, that is, for all $u\in \F_q^*\setminus \cB$ the set $\cB \cup \{u\}$ is not a $(t,\cA,q)$-packing set.
Note that this does not imply that there is no $(t,\cA,q)$-packing set of larger size than $B$ (but such larger sets 
cannot contain $\cB$). 

In Section~\ref{sec:semi} we consider the case~\eqref{eq:special} again and
give a probabilistic construction for a dense sequence of reasonably large $(t,\cA,q)$-packing sets.
This construction is based on some properties of cyclotomic polynomials and resultants. 

Unfortunately there is no efficient decoding procedure for the codes based  
on packing sets. In the case of the sets~\eqref{eq:special} it may be possible to employ some geometry of numbers algorithms, for example for the shortest vector problem (in the $L_\infty$-norm).  
However no precise results or algorithms seem to be known. 

One can also consider generalisations with several sets $\cB_1, \ldots, \cB_k \subseteq \F_q$ of the same cardinality $B$, 
for which the vectors of 
syndromes 
$$
\(S_{\cB_1}(\vec{a}), \ldots, S_{\cB_1}(\vec{a})\), \qquad \vec{a} \in  (\cA\cup \{0\})^{B}, 
$$
are pairwise distinct. The counting arguments of Sections~\ref{sec:upper} and~\ref{sec:lower} extend to 
this case without any difficulties, however we do not see how to generalise the construction of Section~\ref{sec:semi}.

\section{An upper bound}
\label{sec:upper}

Now we present an upper bound on the size of any $(t,\cA,q)$-packing set which can be essentially found in \cite{kl}. We include its proof for the convenience of the reader.
 
\begin{theorem}
\label{thm:upper}
 Let $\cA$ be a subset of $\F_q^*$ of cardinality $A\ge 1$ and let $\cB$ be any $(t,\cA,q)$-packing set  of cardinality $B$. Then
 we have
 $$
\sum_{j=0}^t \binom{B}{j}A^j \le q.
 $$
\end{theorem}
\begin{proof}
Note that the number of possible errors $\vec{a}\in (\cA\cup\{0\})^{B}$ of Hamming weight at most $t$ is
\begin{equation}\label{eq:M}
M=\sum_{j=0}^t \binom{B}{j}A^j,
\end{equation}
which equals the number 
$$
N = \# \left\{ S_{\cB}(\vec{a})~:~\vec{a}\in (\cA\cup\{0\})^{B}\right\}
$$
of corresponding syndromes $S_{\cB}(\vec{a})$,
which are pairwise 
distinct since $\cB$ is a $(t,\cA,q)$-packing set. Using the trivial bound 
$M = N \le q$ we derive the desired bound.
\end{proof} 

\begin{rem}{\rm
It is easy to see that 
$$
\sum_{j=0}^t \binom{B}{j}A^j \ge  \binom{B}{t}A^t > \left(\frac{B-t}{t}\right)^t A^t.
$$
Hence,  Theorem~\ref{thm:upper} implies the bound
 \begin{equation}\label{B}
 B< t\(\frac{q^{1/t}}{A}+1\).
 \end{equation}
}\end{rem}

\begin{example} 
{\rm
 Take $\cA=\F_p^*$, $q=p^B$ and choose $\cB$ to be
  a basis of $\F_q$ over $\F_p$.
 Then each element $x\in \F_q$ has a unique representation 
 $$x=\sum_{b\in \cB} a_b b, \qquad a_b\in \cA\cup \{0\}.
 $$ 
 The number of such elements with at most $t$ nonzero coefficients $a_b$, where $1\le t\le B$, is
 given by $M$ as in~\eqref{eq:M}. Thus, since $\cB$ is a basis of $\F_q$ over $\F_p$,  for $t=B$ we have
 $$
 \sum_{j=0}^{t}  \binom{B}{j} A^j = \sum_{j=0}^B  \binom{B}{j} A^j = (A+1)^B=q
 $$
 and the bound of Theorem~\ref{thm:upper} is attained.
}
\end{example}

\begin{example}  
{\rm Take $q=p$ a prime and $\cA=\{1,2,\ldots,\lambda\}$. 
Set 
$$
L = \fl{\frac{\log p}{\log(\lambda+1)}}.
$$
Then the set 
 $$\cB=\{(\lambda+1)^i: i=0,\ldots,L-1\}$$
 is a $(t,\cA,p)$-packing set for any $t$ with  $2\le t\le L$ since
 the $(\lambda+1)$-adic representation of any integer is unique and the largest exponent~$i$ is chosen such that there is no modulo $p$ reduction. 
 For $t=L$ the bound~\eqref{B} is attained up to a multiplicative constant.
 This is essentially mentioned in \cite{kl}.
}
\end{example}

\section{A lower bound}
\label{sec:lower}
  
  In this section we prove a lower bound on the size of any maximal $(t,\cA,q)$-packing set.

   \begin{theorem}\label{thm:lower}
   Let $\cA$ be a subset of $\F_q^*$ of size $A\ge 1$ and $\cB$ be any maximal $(t,\cA,q)$-packing set of size $B$. Then we have
   $$\sum_{h=0}^t \binom{B+1}{h}A^h \, \sum_{k=1}^t\binom{B}{k-1}A^k +B+1\ge q.$$
  \end{theorem}
  
\begin{proof} Assume $\cB$ is a maximal $(t,\cA,p)$-packing set and $u\not\in \cB$. 
Then for any $u \in  \F_q^*\setminus \cB$ the  set $\cB \cup\{u\}$ is not a  $(t,\cA,p)$-packing set 
and thus we  have 
 $$\sum_{b\in \cB}a_{1,b}b+a_{1,u} u= \sum_{b\in \cB}a_{2,b}b+a_{2,u} u
 $$
 for some 
 $$
 \vec{a}_\nu=(a_{\nu,b})_{b\in \cB\cup\{u\}} \in (\cA\cup\{0\})^{B+1}, \qquad \nu =1,2,
 $$  
 of Hamming weight at most $t$.
 
 In particular, by the definition of a $(t,\cA,q)$-packing set, 
  we have $a_{1,u}\ne a_{2,u}$ and may assume $a_{2,u}\ne 0$. Therefore
 $$
 u=\(a_{1,u} - a_{2,u}\)^{-1}\sum_{b\in \cB}(a_{1,b}-a_{2,b})b.
 $$
 We fix some $h$ and $k$ with $0\le h\le t$ and $1\le k \le t$. Further we assume 
 $$
 w(\vec{a}_1)=h \mand w({\bf a}_2)=k.
 $$
 There are 
 $$\binom{B+1}{h} A^h \mand \binom{B}{k-1}A^k
 $$ 
 choices for   $\vec{a}_1$ and  $\vec{a}_2$, respectively. 
  Hence the number $q-B-1$ of $u\in \F_q^*\setminus \cB$ is bounded by
$$
q-B-1 \le \sum_{h=0}^t \binom{B+1}{h}A^h \, \sum_{k=1}^t  \binom{B}{k-1}A^k
$$
  and the result follows.
\end{proof}

\begin{rem}  
{\rm
   For $t=1$ we get the more precise bound
  \begin{equation}\label{t=1} B\ge \frac{q-1}{\#(\cA/\cA)}\ge \frac{q-1}{A^2},
  \end{equation}
  where $\cA/\cA=\{ab^{-1}: a,b\in \cA\}$ denotes the ratio set of $\cA$, see~\cite[Proposition~2.4]{bur} 
  or~\cite{shrowi}.  Furthermore, in~\cite{shrowi}, an example is given which attains this lower bound
  up to a multiplicative constant. We recall this construction for the convenience of the reader:
  Let $g\in \F_q^*$ be an element of order $k\ge 2$, put $d=\rf{\sqrt{k}}\ge 2$ and choose
  $$\cA=\left\{g,g^2,\ldots,g^d,g^{2d},\ldots,g^{(d-1)d},g^{d^2}\right\}.$$
  Note that $\#\cA\le 2d-1$ and $\cA/\cA$ is the subgroup of $\F_q^*$ of order $k$ generated by $g$.
  Now suppose that $\cB$ is any $(1,\cA,q)$-packing set, that is, 
  $$\cA/\cA\cap \cB/\cB=\{1\}.$$
  Then $\cB$ cannot contain more than one element from each coset of $\cA/\cA$ and thus
  $$B\le (q-1)/k=O\left(q/A^2\right).$$
}
\end{rem}

 \begin{corollary}
 \label{cor:low bound}
   Let $\cA$ be a subset of $\F_q^*$ of size $A$ and $\cB$ be any maximal $(t,\cA,q)$-packing set of size $B$. Then we have
$$B> \frac{(q/(5A))^{1/(2t-1)}}{A}-1.
$$
  \end{corollary}
  
  \begin{proof}
  By~\eqref{t=1} we may assume that $t\ge 2$ and 
  since otherwise the result is trivial we may assume $q\ge 11$ and thus by Theorem~\ref{thm:lower} 
  we have $AB\ge 2$. Now, from Theorem~\ref{thm:lower} and the elementary inequalities 
  \begin{align*}
  \sum_{h=0}^t \binom{B+1}{h}A^h &< \sum_{h=0}^t (A(B+1))^h\\
&   = \frac{(A(B+1))^{t+1}-1}{A(B+1)-1} \le 2\(A(B+1)\)^{t}, \\
 \sum_{k=1}^t  \binom{B}{k-1}A^k &< \frac{1}{B} \sum_{k=1}^t (AB)^k = A\frac{(AB)^t-1}{AB-1} \le  2 A^t B^{t-1}, 
 \end{align*} 
we derive  $q-B-1<  4 A^{2t}(B+1)^{2t-1}$ and the desired bound follows.
\end{proof}

\begin{rem}{\rm
Note that if, for instance, $A< (q/6)^{1/(2t)}$, then the first term in the lower bound of Corollary~\ref{cor:low bound}
dominates and it becomes of order of magnitude $q^{1/(2t-1)}  A^{-2t/(2t-1)}$. 
}\end{rem}

 \section{A Probabilistic Construction}
 \label{sec:semi}
  
 Now we consider the special case~\eqref{eq:special}  and present a probabilistic construction, 
 which for every $\lambda$,  $t$ and a sufficiently large positive $Q$, produces a prime $p \in [Q, 2Q]$
 and  a $(t,\cA,p)$-packing set $\cB$ of large cardinality. 
 
 We describe our construction first. 
 
 \begin{algol} Given arbitrary positive integers $\lambda$  and  sufficiently large positive integers   $K<Q/2$:
\label{alg:cyclot}
\begin{enumerate}
\renewcommand{\labelenumi}{Step~\arabic{enumi}.}

\item Choose a random integer $k \in [K+1, K+K/\log K]$ and test $k$ for primality. 
Repeat this step until a prime $\ell = k$ is found.

\item  Choose a random factored integer $m \in [M,2M]$, where $M = (Q-1)/\ell$, 
and test  $m\ell + 1$ for primality. 
Repeat this step until a prime $p = m\ell + 1$ is found.

\item  Choose a random element $a \in \F_p^*$ and using the knowledge of the factorisation 
of $p-1$ test it  for being a primitive root of $\F_p^*$.  Repeat this step until a primitive 
root $g \in \F_p^*$ is found.
\item Return $b_0 = g^{(p-1)/\ell}$.
\end{enumerate}
\end{algol}

\begin{theorem}\label{thm:constr} Assuming that 
\begin{equation}
\label{eq:KQ}
 \lambda^{2t} (4K)^{2t+2} \log(t \lambda) = o(Q), 
\end{equation}
Algorithm~\ref{alg:cyclot} runs in expected polynomial in $\log Q$  time and with probability 
$1 + o(1)$ returns $b_0 \in \F_p^*$ for which the set   $\cB_0 = \{ b_0,b_0^2, \ldots, b_0^{\ell-1}\} \subseteq \F_p^*$ 
with $\ell = (1+o(1)) K$
is a $(t,\cA,p)$-packing set, where $\cA$ is as in~\eqref{eq:special}.
  \end{theorem}
  
  \begin{proof}   
We first analyse the complexity of Algorithm~\ref{alg:cyclot}
and then show that it is correct with an overwhelming probability.  

\subsubsection*{Running time}  
It follows easily from the classical prime number theorem that  intervals of the form  $[K+1, K+K/\log K]$ contain 
a set $\cL$ of 
\begin{equation}
\label{eq:asymp L}
L =  (1+o(1)) \frac{K}{(\log K)^2} 
\end{equation}
primes, see~\cite[Theorem~10.5]{IwKow}
and the follow-up discussion, which is sufficient for our purpose. In fact the currently strongest result of
Baker, Harman and  Pintz~\cite{BaHaPi} allows to use $k\in[K+1, K+K^\alpha]$ for any fixed  $\alpha >  21/40=0.525$, but this
does not affect our main result.  
Combining this with the  deterministic polynomial time 
primality test of 
Agrawal, Kayal and  Saxena~\cite{AKS} (or with any polynomial time 
probabilistic test, see~\cite {CrPom})  we conclude that Step~1 returns a desired prime~$\ell$ in polynomial time.  

To generate  factored integers in the interval $[M,2M]$ uniformly at random, we use the   polynomial 
time  algorithm of Kalai~\cite{Kal}  which  simplified the previous algorithm of Bach~\cite{Bach}. 
After this we again apply one of the above primality tests to $m\ell+1$. We now need to estimate 
the expected number of such choices before we find a prime 
$$p = \ell m +1.$$

Let, as usual, $\pi(x,k,a)$ denote the number of primes $p\le x$ in the arithmetic progression
$p \equiv a \pmod k$. 

We now recall that since $t \ge 1$ then by~\eqref{eq:KQ} we have $K = O(Q^{1/3})$. Thus by the celebrated Bombieri-Vinogradov theorem, see~\cite[Theorem~17.1]{IwKow}, with the summation extended
only over the primes from the set $\cL$, we immediately derive the following bound 
\begin{equation}
\label{eq:many p}
\pi(2Q,\ell,1)-   \pi(Q,\ell,1) \ge  \frac{Q}{2 \ell \log Q}
\end{equation}
for all but  at most $O\(L/\log Q\)$ primes $\ell \in \cL$. 

Since the primes $\ell$ are generated uniformly at random we see that in expected
polynomial time Step~2 outputs the desired prime $p$.

Since $\ell$ is prime and the prime number factorisation of $m$ is known, one can test 
whether $a \in \F_p^*$ is a primitive root of $\F_p^*$ in deterministic polynomial time.  
Recall that the density of primitive roots in $\F_p^*$ is high enough, 
namely, it is 
$$ 
 \frac{\varphi(p-1)}{p-1} \ge c \frac{1}{\log \log p}
$$ 
for an absolute constant $c>0$,
which easily follows from Mertens' formula, see~\cite[Equation~(2.16)]{IwKow}. We 
now immediately conclude that Step~3 outputs a primitive root of~$\F_p^*$ in 
 expected polynomial time. 

The complexity analysis of  Step~4 is trivial. 

\subsubsection*{Correctness}  Take $\cB_0 = \{  b_0,b_0^2, \ldots, b_0^{\ell-1}\} \subseteq \F_p^*$. 
 If 
 \begin{equation}
 \label{fail} 
 \sum_{b\in \cB_0}a_b b\not= 0
 \end{equation}
 for all $\vec{a}=(a_b)_{b\in \cB_0}\in \{-\lambda,-\lambda+1,\ldots,\lambda\}^{\ell-1}$ with $1\le w(\vec{a})\le 2t$,
 then $\cB_0$ is a $(t,\cA,p)$-packing set.

 Let $\Phi_\ell(X)\in \Z[X]$ be the $\ell$th cyclotomic polynomial which completely splits over $\F_p$, that is, 
 $\cB_0$ is exactly its set of zeros.

 If~\eqref{fail} fails,  then there is a nonzero  polynomial with at most $2t$ nonzero coefficients 
 $a_i\in \{-\lambda,\ldots,\lambda\}$ of the form
 \begin{equation}
 \label{eq:polys} 
 f(X)=\sum_ {i=0}^{\ell-2} a_i X^i\in \F_p[X]
 \end{equation}
 of degree at most $\ell-2$
 which vanishes at some  $b \in \cB_0$. 
 The resultant $R_f$ of $f(X)$ and $\Phi_\ell(X)$ is
 $$R_f=\prod_{\xi:~\Phi_\ell(\xi)=0}f(\xi),$$
 where the product is taken over all complex primitive $\ell$th roots of unity~$\xi$, 
see, for example,~\cite[Theorem~1]{bili10}, and
 vanishes modulo $p$:
 \begin{equation}
 \label{eq:Res div} 
 R_f \equiv 0 \pmod p. 
 \end{equation}

 Consider the set of
 \begin{equation}
 \label{eq:Bound N} 
 N=\sum_{i=1}^{2t} \binom{\ell-1}{i} (2\lambda)^i\le 2(2 \ell \lambda)^{2t}  
 \end{equation} 
 different polynomials of the form~\eqref{eq:polys}  with at most $2t$ nonzero 
 coefficients $a_i\in \{-\lambda,\ldots,\lambda\}$.
Since $\Phi_\ell(X)$ is irreducible over $\Q$ and for any $f(X)$  of the form~\eqref{eq:polys}  
we have  $\deg(f)\le \ell -2 < \deg(\Phi_\ell)$,  
 the $N$ resultants $R_f$ do not vanish over $\Q$ and their size $|R_f|$ is bounded by 
 $$
 |R_f|= \prod_{\xi:~\Phi_\ell(\xi)=0}|f(\xi)|\le (2t\lambda)^{\ell-1}.$$
So each $R_f$ has at most 
 \begin{equation}
 \label{eq:Bound R}
O\(\log  |R_f|\) = O\(\ell \log(t\lambda)\) 
 \end{equation}
prime divisors.
 Hence,  from~\eqref{eq:Bound N} and~\eqref{eq:Bound R}  we see that the set $\cS$ of primes $p$  that  satisfy~\eqref{eq:Res div} for at least one of the resultants 
 $R_f$ is of cardinality
 \begin{equation}
 \label{eq:Bound S}
S =   O\(\lambda^{2t} (2 \ell )^{2t+1} \log(t \lambda)\).
\end{equation}
On the other hand, we see from~\eqref{eq:asymp L} and~\eqref{eq:many p} that 
Algorithm~\ref{alg:cyclot} produces integers $m$ uniformly at random from a set of 
cardinality  
\begin{equation}
 \begin{split}
 \label{eq:Bound M}
 M +O(1)=    Q/\ell +O(1).
\end{split}
\end{equation}

Comparing~\eqref{eq:Bound S} with~\eqref{eq:Bound M}  we see that under the condition~\eqref{eq:KQ} we have
$S = o(M)$  which concludes the proof.    
 \end{proof} 
 
 \begin{rem}{\rm
Note that a subgroup $\cB$ of $\F_p^*$ of order $\ell$ is a $(1,\cA,p)$-packing set whenever $\cA$ contains at most one element from each coset of $\cB$. For example, if $\ell=(p-1)/2$, that is, 
  $\cB$ is 
  the subgroup of quadratic residues modulo $p$, then $\cB$ is a $(1,\{1,2\},p)$-packing set whenever $p\equiv \pm 3\bmod 8$, that is, whenever $2$ is a quadratic nonresidue modulo $p$.
Furthermore, for $t\ge 2$, recent advances towards the Waring problem in $\F_p$, 
see, for example,~\cite{CCP, CoPi,CHPS},  imply rather 
severe restrictions on the order $\ell$ of a subgroup $\cB$  of $\F_p^*$, 
  for which $\cB$ can be
 a $(t,\{1\},p)$-packing set. 
      }\end{rem}

\section*{Acknowledgements}

This work was initiated during a visit of the authors to the Paris Lodron University of Salzburg.
The authors are grateful for its support. 

During the preparation of this work the first author was  supported   by the ARC Grants~DP170100786 and DP180100201, 
and the second author was supported by the Austrian Science Fund FWF Project P~30405-N32. 
 
 The authors like to thank the anonymous referees for their valuable suggestions.


\begin{thebibliography}{99} 
 
 \bibitem{AKS}  M. Agrawal, N. Kayal and N. Saxena, 
 `PRIMES is in P', {\it  Ann. of Math.\/}, {\bf 160}  (2004),   781--793.

 \bibitem{ah} R. Ahlswede, H. Aydinia and L. Khachatrian, 'Unidirectional error
control codes and related combinatorial problems',
{\it Proc. Eighth
Int. Workshop Algebr. Combin. Coding Theory (ACCT-8)\/}, St. Petersburg, Russia, 2002,   6--9.

\bibitem{Bach} E. Bach, `How to generate factored random numbers',  
{\it SIAM J. Comp.\/}, {\bf 17} (1988), 179--193.

\bibitem{BaHaPi}
R. C. Baker, G. Harman and J. Pintz,
`The difference between consecutive primes. II',
{\it Proc. Lond. Math. Soc.\/},  {\bf 83} (2001), 532--562.
 
 \bibitem{be} E. Berlekamp, {\it Algebraic coding theory\/}. Revised ed., World Scientific Publishing Co. Pte. Ltd., Hackensack, NJ, 2015.
 
 \bibitem{bili10} Y. Bistritz and A. Lifshitz, 'Bounds for resultants of univariate and bivariate polynomials', {\it Linear Algebra Appl.\/}, {\bf 432} (2010), 1995--2005.
 
 \bibitem{bur} M. Buratti, `Packing the blocks of a regular structure', {\it Bull. Inst. Combin. Appl.\/}, {\bf 21} (1997), 49--58.
 
 \bibitem{ca} Y. Cassuto, M. Schwartz, V. Bohossian and J. Bruck, 
`Codes for asymmetric limited-magnitude errors with application to multilevel flash memories', {\it IEEE Trans. Inform. Theory\/}, {\bf 56} (2010),  1582--1595. 
 
 \bibitem{CCP} J. Cipra, T. Cochrane and C. G. Pinner, 
`Heilbronn's Conjecture on Waring's number $\bmod\, p$', 
{\it J. Number Theory\/}, {\bf  125} (2007), 289--297.



\bibitem{CHPS}  T. Cochrane, D. Hart,  C. Pinner and C. Spencer,
`Waring's number for large subgroups of $\Z_p$', 
{\it Acta Arith.\/}, {\bf  163} (2014),  309--325.  
 
 \bibitem{CoPi}  T. Cochrane and C. Pinner,
`Sum-product estimates applied to Waring's problem   $\bmod\, p$', 
{\it Integers\/}, {\bf 8} (2008),   A46, 1--18. 

 \bibitem{CrPom} R.~Crandall and C.~Pomerance, 
{\it Prime numbers: A computational perspective\/}, 
 Springer-Verlag, New York, 2005.
 
\bibitem{IwKow} H. Iwaniec and E. Kowalski,
{\it Analytic number theory\/}, Amer.  Math.  Soc.,
Providence, RI, 2004.

\bibitem{Kal} A. T. Kalai,
`Generating random factored numbers, easily',
{\it J.  Crypto.\/}, {\bf 16} (2003), 287--289.
 
 \bibitem{kl} T. Kl\o ve, J. Luo, I. Naydenova and S. Yari,  `Some codes correcting asymmetric errors of limited magnitude', {\it IEEE Trans. Inform. Theory\/}, {\bf 57} (2011),   7459--7472. 
 
 \bibitem{lixi} S. Ling and C. Xing, {\it Coding theory. A first course\/}, Cambridge Univ. Press, Cambridge, 2004.
 
 \bibitem{li} J. H. van Lint, {\it Introduction to coding theory\/}, 3rd edition. Graduate Texts in Mathematics, vol.86. Springer-Verlag, Berlin, 1999. 
 
 \bibitem{mc} F. J. MacWilliams and N. J. A. Sloane, {\it The theory of error-correcting codes, Vol.~II\/}, North-Holland Math.  Library, vol.~16. North-Holland Publishing Co., Amsterdam -- New York -- Oxford, 1977.
 
   \bibitem{shrowi} O. Roche-Newton, I. D. Shkredov and A. Winterhof, `Packing sets over finite abelian groups', {\it Integers\/},  {\bf 18}, A38, 1--9.
 
  \bibitem{var} R. R. Varshamov, `A class of codes for asymmetric channels and a problem from the additive theory of numbers', {\it IEEE Trans. Inform.  Theory\/}, {\bf 19} (1973),   92--95.

  \bibitem{kl2}  S. Yari, T. Kl\o ve and B. Bose, `Some codes correcting unbalanced errors of limited magnitude for flash memories', {\it IEEE Trans. Inform. Theory\/}, {\bf 59} (2013),  7278--7287.

\end{thebibliography}
\end{document}